\documentclass[aps,pra,preprint,groupedaddress, showkeys]{revtex4-2}

\usepackage{algorithm}
\usepackage{algpseudocode}
\usepackage{amsmath}
\usepackage{amsthm}
\usepackage{amssymb}
\usepackage{bbold}
\usepackage{cleveref}
\usepackage{enumitem}  
\usepackage{float}
\usepackage[utf8]{inputenc}
\usepackage{mathrsfs}
\usepackage{mathtools}
\usepackage[nice]{nicefrac}
\usepackage{tabularx}
\usepackage[normalem]{ulem}
\usepackage{verbatim}
\usepackage{xcolor}

\newtheorem{thm}{Theorem}
\newtheorem{defo}[thm]{Definition}
\newtheorem{lem}[thm]{Lemma}

\newcommand{\anc}{\texttt{anc}}
\newcommand{\bb}{\mathbf{B}}
\newcommand{\bra}[1]{\langle #1|}
\newcommand{\braket}[2]{\langle #1|#2\rangle}
\newcommand{\cA}{{\cal A}}
\newcommand{\cB}{{\cal B}}

\newcommand{\cH}{{\cal H}}
\newcommand{\cM}{{\cal M}}
\newcommand{\cN}{{\cal N}}

\newcommand{\da}{^\dagger}

\newcommand{\ket}[1]{|#1\rangle}
\newcommand{\ketbra}[1]{\ket{#1}\bra{#1}}
\newcommand{\lm}{\lambda_{\max}}
\newcommand{\om}{\omega}
\newcommand{\sep}[1][]{{\texttt{SEP}}_{#1}}
\newcommand{\Tr}[1][]{{\rm Tr}_{#1}}
\newcommand{\1}{\mathbb{1}}
\renewcommand{\leq}{\leqslant}
\renewcommand{\geq}{\geqslant}

\floatname{algorithm}{Protocol}

\crefname{algorithm}{protocol}{protocols}
\Crefname{algorithm}{Protocol}{Protocols}

\begin{document}

\title{Honest-binding quantum bit commitment from separable operations}

\author{Ziad Chaoui, Anna Pappa}
\affiliation{Technische Universität Berlin, Germany}
\author{Matteo Rosati}
\affiliation{Università degli Studi Roma Tre, Italy}

\date{\today}

\begin{abstract}
Bit commitment is a fundamental cryptographic primitive and a cornerstone for numerous two-party cryptographic protocols, including zero-knowledge proofs. However, it has been proven that unconditionally secure bit commitment, both classical and quantum, is impossible. In this work, we demonstrate that imposing a restriction on the committing party to perform only separable operations enables secure quantum bit commitment schemes. Specifically, we prove that in any perfectly hiding bit commitment protocol, an honestly-committing party limited to separable operations will be detected with high probability if they attempt to alter their commitment. To illustrate our findings, we present an example protocol.
\end{abstract}

\keywords{quantum information, quantum cryptography, quantum bit commitment}

\maketitle

\section{Introduction}
Bit commitment is a cryptographic primitive between two mistrustful parties that allows the committer (Alice) to commit to a bit and reveal it to another party (Bob) at a later stage. More specifically, a bit commitment protocol consists of two phases. In the commit phase, Alice secretly chooses a bit $b\in\{0,1\}$ and sends Bob proof of her commitment without revealing the bit. In the opening phase, Alice announces her bit and Bob verifies that it is indeed the bit she had committed to. Informally, it is required that a bit commitment protocol is \emph{binding} (i.e. Alice cannot change her bit after the commit phase) and at the same time \emph{hiding} (i.e. Bob does not learn the bit before the opening phase).

Quantum bit commitment is a fundamental primitive in quantum cryptography. Indeed, it can be used to build quantum oblivious transfer \cite{brassard_quantum_1993}, and consequently any 2-party computation. However, Mayers \cite{mayers_unconditionally_1997} and independently Lo and Chau \cite{lo_why_1998} proved that unconditionally secure quantum bit commitment is not possible, thus bringing into question the conditions under which it can actually be realised. A lot of research has therefore been focused on restricting the power of the adversaries involved in the quantum bit commitment protocol. The most common examples include putting bounds on their quantum memory, both in size \cite{bounded} and quality \cite{koenig_unconditional_2012}, and also considering (relativistic) constraints on the location of the parties \cite{kaniewski_secure_2013}. 
{ {
Relatively unexplored are models where restrictions are placed on the operations accessible to the parties, motivated by realistic hardware limitations. One such model was proposed in \cite{salvail_quantum_1998}, where the adversary, i.e. a dishonest Alice, has access to a large quantum memory, but cannot act upon its entirety in a coherent fashion. More precisely, Alice has $n$ qubits, but can only act on up to $k\leqslant n$ of them at a time. These operations can be modelled as local operations and classical communication (LOCC) channels. LOCC channels are indeed a natural way to model realistic hardware assumptions; they precisely allow us to represent that in practice quantum devices rarely allow large, joint coherent operations across many qubits, as is the case for example in modular (see \Cref{fig:modular}) and distributed quantum computing setups  \cite{carreavazquez24,Main2025}.
However, the security proof of the proposed protocol in \cite{salvail_quantum_1998} was erroneous \cite{funder}, and therefore the question remained whether secure bit commitment protocols exist in such models.
\begin{figure}[!t]
    \includegraphics[scale=0.1]{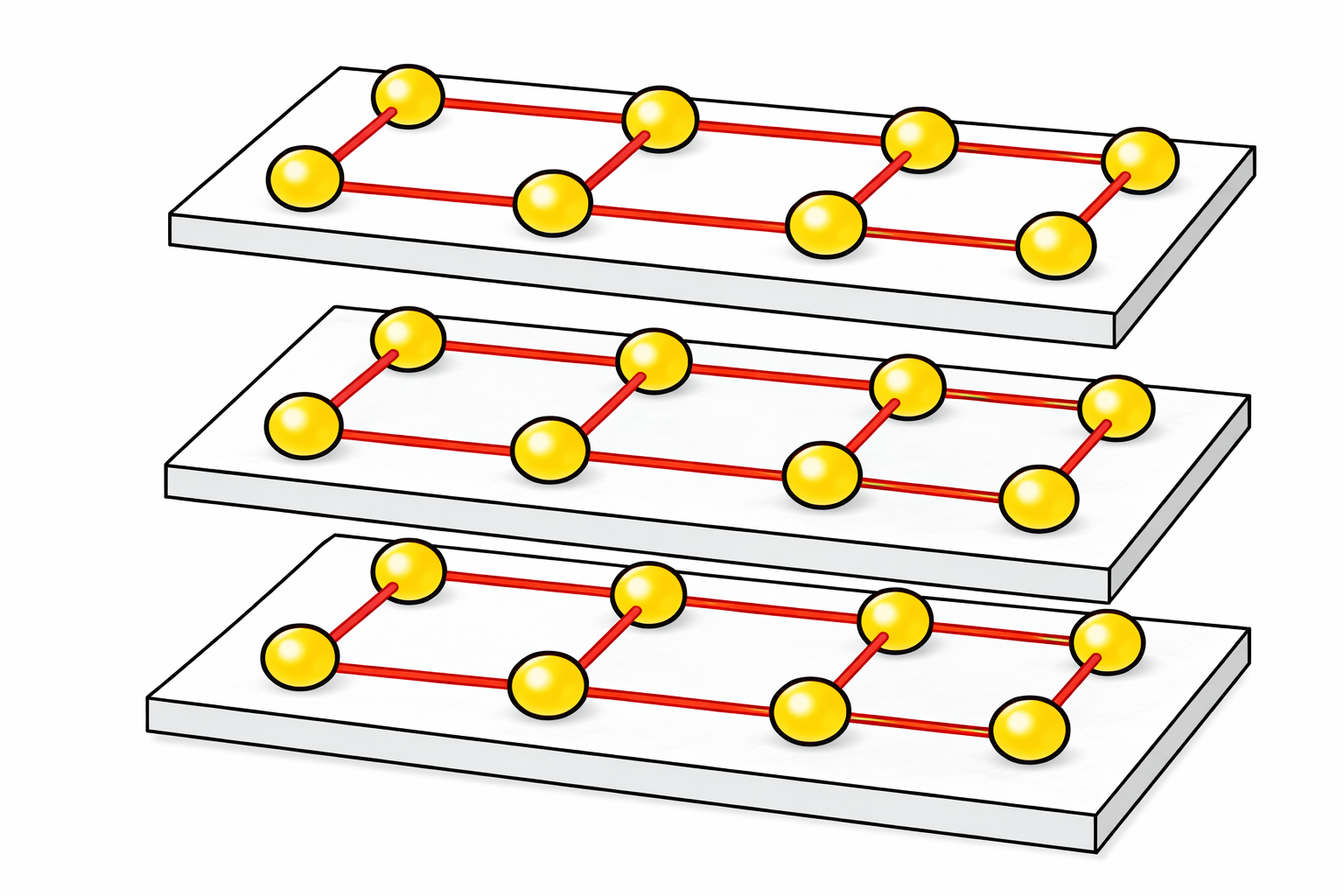}
    \caption{\textbf{Schematic representation of a modular quantum processor}. The system is composed of multiple quantum chips arranged in a stack, where coherent measurement operations are restricted to individual chips rather than across the full device.}
    \label{fig:modular}
\end{figure}

In this work, we present a set of conditions that, if fulfilled after the commit phase, allows for secure quantum bit commitment under the restriction that the committing party (Alice) can only perform separable operations on her quantum registers. This restriction includes the restriction proposed in \cite{salvail_quantum_1998}, since  LOCC channels are a subset of the more general separable channels \cite{chitambar_locc}. }}Further, we propose a protocol in which said conditions are satisfied, and show that Alice's attempts to switch commitment after the commit phase will be detected with high probability. 

We first define what it means for a bit commitment scheme to be secure. We then study the properties of the shared states following Alice's commitment, focusing on scenarios where her reduced state is either a product state or an absolutely maximally entangled state, depending on whether she committed to $0$ or $1$, respectively. In these settings, we prove that restricting her actions to separable operations leads to security. Finally, we present a bit commitment protocol where Alice has a 3-qudit register and can only operate on two of these qudits. We prove that she is not able to change her commitment except with probability $\nicefrac{1}{d}$, where $d$ is the dimension of the qudits used to communicate.

\section{Security definitions}
As previously mentioned, a bit commitment protocol is desired to be at the same time hiding and binding. The definition of the first property is quite straightforward.
{ 
\begin{defo}[$\mu$-hiding]\label{def:hiding}
    A quantum bit commitment protocol is \emph{$\mu$-hiding}, with $\mu\geq 0$, if Bob can guess the committed bit before the opening phase only with probability $\nicefrac{1}{2}+\mu$. It is said to be \emph{perfectly hiding} if $\mu=0$, i.e.~Bob can only hazard a random guess.
\end{defo} 
}
On the other hand, there are different ways to define what it means for a protocol to be binding. There is a composable definition given in \cite{schaffner_boundedstorage}, while the most common definition was first presented in \cite{goos_perfectly_2000}. In the latter, the authors define a bit commitment protocol to be binding if Alice cannot significantly increase the total probability of successfully opening both $0$ and $1$ after the commit phase. In \cite{schaffner_boundedstorage} this is referred to as \emph{weak-binding} and authors in \cite{agrawal_general_2022} and \cite{fang_et_al:LIPIcs.ISAAC.2022.26} call this property \emph{sum-binding}, while also defining the notion of  \emph{honest-binding}. Informally, a protocol is honest-binding if Alice cannot successfully open the bit $1-b$ after she honestly committed to the bit $b$. We formally define sum-binding and honest-binding as follows:
{ 
\begin{defo}[$\delta$-sum-binding]
    A bit commitment protocol is \emph{$\delta$-sum-binding} if for any strategy that Alice follows, the probabilities $p_b$ of successfully opening the bit $b\in\{0,1\}$ satisfy
    \begin{equation*}
        p_0+p_1\leq 1+\delta.
    \end{equation*}
\end{defo}
}
\begin{defo}[$\varepsilon$-honest-binding]\label{def:honest_binding}
    In a bit commitment protocol, for $b\in\{0,1\}$, let $p_b$ be the probability that Alice successfully opens the bit $1-b$ after having honestly committed the bit $b$. We say the protocol is \emph{$\varepsilon$-honest-binding} if:
    \begin{equation*}
        \max\{p_0,p_1\}\leq \varepsilon.
    \end{equation*}
\end{defo}

Although sum-binding implies honest-binding it is still interesting to study the latter because it is precisely the honest-binding property that fails in the unconditional setting, leading to the no-go result \cite{mayers_unconditionally_1997,lo_why_1998}. Gaining a better understanding of how to evade this no-go result, can help design protocols that are conditionally binding and hiding. In the following, we show that restricting Alice to separable operations can guarantee honest-binding, without comprising perfect hiding. We first take a closer look at properties of Alice and Bob's shared state after Alice committed honestly.

\section{Binding from separable operations}
\subsection{Post-commitment states}
Let $\cA$ be the Hilbert space of Alice and $\cB$ be the Hilbert space of Bob, associated with an $n$-qudit and an $m$-qudit quantum register respectively. We set $N \coloneqq \dim(\cA)=d^n$ and $M\coloneqq \dim(\cB)=d^m$, where $d\geq 2$ is the dimension of the qudits.
After the commit phase, Alice and Bob share one of two quantum states $\Psi_0,\Psi_1\in\cA\otimes\cB$ depending on  Alice's commitment to $0$ or $1$ respectively. We assume the states to be pure~\footnote{If the states are not pure we can purify them over a larger system that will be Bob's, since the size of his system does not matter.}, and we write $\Psi_b=\ketbra{\Psi_b}$ for $b\in\{0,1\}$. For the protocol to be hiding, the reduced states on Bob's side have to be indistinguishable~\footnote{This condition is necessary, but not sufficient. The analysis of the hiding property is protocol-specific as is the case for our protocol in the next section.}. This means that:
\begin{equation}\label{eq:equiv_purifications}
    \Tr[\cA](\ketbra{\Psi_0})=\Tr[\cA](\ketbra{\Psi_1}).
\end{equation}

We study the situation where Alice and Bob share the post-commitment state $\Psi_b$ and Alice wants to open $1-b$, meaning she wants to convince Bob that they are sharing the state $\Psi_{1-b}$. We denote with $p_b$ the probability of her successfully doing so, i.e., Bob accepting her opening.
To this end Alice performs an operation on her space $\cA$, described by $\cN\in C(\cA)$, where $C(\cA)$ is the set of completely positive trace-preserving maps, i.e. quantum channels, from the space of linear operators on $\cA$ to itself.
Her goal is to maximise $p_b$, the probability that Bob identifies the shared state as $\Psi_{1-b}$, after she applied her channel to the committed state $\Psi_b$; 
this is expressed by the fidelity between the two states \footnote{For two quantum states $\rho$ and $\sigma$, we define the fidelity as
$
F(\rho,\sigma):=\|\sqrt{\rho}\sqrt{\sigma}\|_1^2=\left(\Tr(\sqrt{\sqrt{\sigma}\rho\sqrt{\sigma }})\right)^2
$}:
\begin{equation*}
p_b= \max_{\cN\in C(\cA)}F(\ketbra{\Psi_{1-b}},(\cN\otimes\1_\cB)(\ketbra{\Psi_b})).
\end{equation*}
We fix a channel $\cN$ and denote by $\{K_j\}$ its set of Kraus operators, i.e.~$\cN(X)=\sum_jK_jXK_j\da$. Since $\Psi_{1-b}$ is pure, it follows that:
\begin{equation}\label{eq:fidelity_overlapsum}
     F(\ketbra{\Psi_{1-b}},(\cN\otimes\1_\cB)(\ketbra{\Psi_b}))= \sum_j|\bra{\Psi_{1-b}}(K_j\otimes\1_\cB)\ket{\Psi_b}|^2.
\end{equation}
We now take a closer look at the states $\Psi_0$ and $\Psi_1$. From \eqref{eq:equiv_purifications}, it follows that there exist orthonormal sets $\{\ket{x_i}\}_{i=0}^{N-1}$, $\{\ket{y_i}\}_{i=0}^{N-1}\subset \cA$ and  $\{\ket{b_i}\}_{i=0}^{M-1}\subset\cB$, as well as $\lambda_i\geq 0$, so that we can write the Schmidt decompositions \cite{nielsen_quantum_2012} of $\ket{\Psi_0}$ and $\ket{\Psi_1}$ as
\begin{align}
    \ket{\Psi_0}&=\sum_{i=0}^{M-1}\sqrt{\lambda_i}\ket{x_i}\ket{b_i},\label{eq:schmidt_0}\\
    \ket{\Psi_1}&=\sum_{i=0}^{M-1}\sqrt{\lambda_i}\ket{y_i}\ket{b_i}\label{eq:schmidt_1}.
\end{align}
We can now compute (see appendix \ref{apx:gen} for details)
\begin{align}
    p_0&= \sum_j\left|\sum_{i} \lambda_i\bra{y_i}K_j\ket{x_i}\right|^2\label{eq:switch_schmidt_0}\\
    p_1&=\sum_j\left|\sum_{i} \lambda_i\bra{x_i}K_j\ket{y_i}\right|^2\label{eq:switch_schmidt_1}.
\end{align}

By the unitary equivalence of purifications, there exists a unitary operator $U$  such that $U\ket{x_i}=\ket{y_i}$, for every $i$. Therefore, if Alice's operations are unrestricted, she can choose the unitary channel with the single Kraus operator $U$ and perfectly cheat without being detected; this is exactly the famous no-go result \cite{mayers_unconditionally_1997,lo_why_1998}. We can however avert this by restricting Alice to performing only separable channels, as we show in the next section. 

\subsection{Binding from separable operations}
 A separable channel is defined with respect to a bipartition of Alice's space. A separable channel $\cN\in C(\cA)$ with respect to the bipartition $\cA=\cA_1\otimes\cA_2$, is a convex combination of product maps $\cN_{\cA_1}\otimes\cN_{\cA_2}$, where the $\cN_{\cA_i}$ are completely  positive maps on $\cA_i$ and the convex combination is also trace preserving. Let $\cA=\cA_1\otimes\cA_2$ be such a  bipartition and denote $N_1\coloneq\dim(\cA_1)$ and $N_2\coloneq\dim(\cA_2)$ (without loss of generality, we assume that $N_2\leq N_1$).   $\sep(\cA_1,\cA_2)$ is the set of separable channels with respect to this bipartition. An important feature of separable channels is that their Kraus operators are also separable with respect to the same partition. For a separable channel $\cN\in \sep(\cA_1,\cA_2)$ we can write its Kraus operators as $K_j=K_{j1}\otimes K_{j2}$.

In line with the notation chosen for the Schmidt decomposition \eqref{eq:schmidt_0} and \eqref{eq:schmidt_1}, we investigate the situation where $\ket{x_i}$ are maximally entangled orthonormal states and $\ket{y_i}$ orthonormal product states (so that Alice's reduced state is separable). 
This means that there exist unit vectors $\ket{v_i}\in\cA_1$ and $\ket{w_i}\in\cA_2$, such that: \begin{equation}\label{eq:product_state}
\ket{y_i}=\ket{v_i}\ket{w_i}.
\end{equation}
We can write
    $K_j\ket{y_i}=(K_{j1}\otimes K_{j2})(\ket{v_i}\ket{w_i})=\ket{v_{ji}}\ket{w_{ji}}$, 
where $\ket{v_{ji}}\in\cA_1$ and $\ket{w_{ji}}\in\cA_2$ are of norm less than $1$.

In contrast, the $\ket{x_i}$ are maximally entangled with respect to the bipartition $\cA_1\otimes\cA_2$. This implies that, for every $i$, there are bases $\{\ket{e_{ik}}\}_{k=0}^{N_1-1}\subset\cA_1$ and $\{\ket{f_{ik}}\}_{k=0}^{N_2-1}\subset\cA_2$, such that:
\begin{equation}\label{eq:entangled_states}
    \ket{x_i}=\frac{1}{\sqrt{N_2}}\sum_{k=0}^{N_2-1}\ket{e_{ik}}\ket{f_{ik}}.
\end{equation}
With this we find the following upper bounds for $p_0$ and $p_1$ (see appendix \ref{apx:gen} for detailed calculations)
\begin{align}
    p_0&\leq \lm^2 N_2, \label{eq:alice_switch_bound_0}\\
    p_1&\leq \frac{1}{N_2} \label{eq:alice_switch_bound_1},
\end{align}
 which depend on the maximum eigenvalue of $\Psi_0$ and on $N_2$, the dimension of the \emph{smaller} of the two bipartitions of Alice's space.
 
If Bob can now choose $M=N_2$ and $\lambda_i=\nicefrac{1}{N_2}$ for every $i$, then the states in \eqref{eq:schmidt_0} and \eqref{eq:schmidt_1} are both uniform superpositions. It follows from \eqref{eq:alice_switch_bound_0} and \eqref{eq:alice_switch_bound_1} that Alice's cheating probability is bounded by $\nicefrac{1}{N_2}$. We stress that this can be achieved under the motivated assumption that it is Bob who creates the states and therefore fixes $M$ and $\lm$; indeed, if we allowed Alice to create an entangled state while being restricted to separable operations, that would be conceptually a contradiction regarding her capabilities.
Furthermore, we note that Bob could also choose $M>N_2$ and reduce $p_0$ even more, but this would not affect Alice's overall cheating probability.

On the other hand, Alice can freely choose the particular bipartition $\cA_1\otimes\cA_2$ on which to act with separable operations, whereas \eqref{eq:product_state} and \eqref{eq:entangled_states} only guarantee that $\ket{x_i}$ is maximally entangled and $\ket{y_i}$ separable with respect to this particular bipartition of $\cA$. To ensure the binding property, we then need to consider states that fit the structures \eqref{eq:product_state} and \eqref{eq:entangled_states} with respect to \emph{any} bipartition of $\cA$.
We want the $\ket{y_i}$'s to be separable with respect to any bipartition of $\cA$, and will therefore consider the product states 
$\ket{y_i}=\ket{a_{i1}}\otimes\cdots\otimes\ket{a_{in}}$,
where, for every $j$, $\ket{a_{ij}}$ is a unit vector in the $d$-dimensional Hilbert space $\cH_d$.
As for the $\ket{x_i}$'s, the states that are maximally entangled with respect to any bipartition of $\cA$ are precisely the absolutely maximally entangled (AME) states:

\begin{defo}[Absolutely maximally entangled states \cite{helwig_absolute_2012}]\label{def:ame}

An $AME(n,d)$ state of $n$ qudits of dimension $d$,  is a pure state $\ket{\gamma}\in(\cH_d)^{\otimes n}=\cA$, such that for any bipartition of the system $\cA=\cA_1\otimes\cA_2$, with $N_2=\dim(\cA_2)\leq \dim(\cA_1)=N_1$, $\ket{\gamma}$ is maximally entangled, i.e.
$\Tr[\cA_1](\ketbra{\gamma})=\frac{1}{N_2}\1_{N_2}$.\\
\end{defo}
\vspace{-0.8cm}
\noindent We are now equipped to prove the following Theorem:
\begin{thm}\label{thm:main_result}
    We consider a bit commitment scheme between Alice and Bob, who hold an $n$-qudit register $\cA$ and an $m$-qudit register $\cB$ respectively, with $d$ the qudit dimension. If after the commit phase Alice and Bob either share the state $\Psi_0$ or $\Psi_1$ of the form \eqref{eq:schmidt_0} or \eqref{eq:schmidt_1} respectively, then the bit commitment protocol is $\nicefrac{1}{d}$-binding under the following two conditions:
    \begin{enumerate}
        \item Alice can only perform separable channels on $\cA$.
        \item In line with notation in \eqref{eq:schmidt_0} and \eqref{eq:schmidt_1} for every $i\in\{0,\dots, N-1\}$ it holds that $\ket{x_i}\in AME(n,d)$ and $\ket{y_i}$ is a separable state with respect to any bipartition of $\cA$\label{cond_reduced_state}.
    \end{enumerate}
    \label{thm_bin}
\end{thm}

\begin{proof}
The two conditions given in the theorem imply that the analysis presented so far holds for \emph{any} bipartition of Alice's space $\cA$. As discussed above Bob creates the states and therefore can freely choose $\lm$ and $M$ in the Schmidt decompositions \eqref{eq:schmidt_0} and \eqref{eq:schmidt_1}. From \eqref{eq:alice_switch_bound_0} and \eqref{eq:alice_switch_bound_1} it then follows that Alice's maximal cheating probability is 
    \begin{align*}
\max_{N_2}\min_{\lm}\max\left\{\frac{1}{N_2}, \lm^2 N_2\right\}.
\end{align*}
 Alice can maximise $\nicefrac{1}{N_2}$ by choosing the smallest possible $N_2$ which corresponds to her acting on a single qudit, obtaining $\max_{N_2}\nicefrac{1}{N_2}=\nicefrac{1}{d}$. If she wants to maximise $\lm^2 N_2$ she can pick $N_2$ to be as large as possible, i.e.~$\lfloor\nicefrac{N}{2}\rfloor$ (because $\cA_2$ is the smaller partition), but since Bob has control over $\lm$ he can choose it to be $\lfloor\nicefrac{N}{2}\rfloor^{-1}$. Alice acting on a partition of dimension $\lfloor\nicefrac{N}{2}\rfloor$ corresponds to her operating on $\lfloor\nicefrac{n}{2}\rfloor$ of her qudits, and we can conclude that her maximum cheating probability is then
 \begin{equation*}
    p_{switch}=\max\left\{\frac{1}{d},\frac{1}{d^{\lfloor\frac{n}{2}\rfloor}}\right\}=\frac{1}{d}.
\end{equation*}
\end{proof}

Therefore, we are left with the challenge of designing a quantum bit commitment protocol where, after commitment,  Alice's reduced state is either a product state or an AME state depending on her commitment. If there is a protocol where Alice's reduced state satisfies this condition regardless of her strategy up to the end of the commit phase, then our analysis proves that the protocol is sum-binding. If the protocol leads to the desired situation only if Alice is honest up to the commitment, then our analysis proves that the protocol is honest-binding. Crucially, in either case the protocol is perfectly hiding, since the post-commitment states satisfy \eqref{eq:schmidt_0} and \eqref{eq:schmidt_1}. In the following we present an example of an honest-binding protocol for any $d$ using $AME(3,d)$-states, where Alice has a 3-qudit register and is restricted to perform coherent operations on up to 2 of them.
\pagebreak
\section{AME(3,d) quantum bit commitment}\label{sec:ame3d}

We consider the following set up: Alice has a $3$-qudit register $\cA$, Bob a single-qudit register $\cB$,
and we also consider ancilla registers $\anc$. 
The Fourier gate for qudits is
$F=\frac{1}{\sqrt{d}}\sum_{k=0}^{d-1}\om^{kl}\ket{k}\bra{l}$, where $\om:=e^{\frac{2\pi i}{d}}$ is the primitive $d$-th root of unity; the gate transforms the generalized $Z$-eigenbasis $\mathbf{B}_1=\{\ket{k}\}_{k=0}^{d-1}$ to the generalized $X$-eigenbasis  $\mathbf{B}_0=\{\ket{\tilde k}\}_{k=0}^{d-1}$ , with
$\ket{\tilde k}=F\da\ket{k}=\frac{1}{\sqrt{d}}\sum_{l=0}^{d-1}\om^{-kl}\ket{l}$ \cite{nielsen_quantum_2012}. We can now state the protocol and prove its security.

\begin{algorithm}[H]
\algrenewcommand\algorithmicrequire{\textbf{Commit Phase}}
\algrenewcommand\algorithmicensure{\textbf{Opening Phase}}
\caption{AME(3,d) Quantum Bit Commitment}\label{ame3d-qbc}
\vspace{\baselineskip}
\begin{algorithmic}[1]
\State Bob creates the state \begin{equation}
        \ket{\Xi}=\frac{1}{\sqrt{d}}\sum_{l=0}^{d-1}\ket{l}_\anc\ket{\Phi_l}_{\cA\cB},
    \end{equation} 
    where
    \begin{equation}
    \ket{\Phi_l}=\frac{1}{\sqrt{d}}\sum_{j=0}^{d-1}\om^{jl}\ket{jjj}_\cA\ket{j+l}_\cB.
    \end{equation}\label{def:phi_l}
\Require
\State Bob sends Alice the registers $\anc$ and $\cA$.
\State Alice picks a permutation $\pi\in S_d$ and chooses $b\in\{0,1\}$ she wants to commit.
    \begin{enumerate}
        \item To commit $b=0$ Alice measures $\anc$ in the basis $\bb_{0,\pi}=\{\ket{\widetilde{\pi(k)}}\}_{k=0}^{d-1}$.
        \item To commit $b=1$ Alice measures $\anc$ in the basis $\bb_{1,\pi}=\{\ket{\pi(k)}\}_{k=0}^{d-1}$.
    \end{enumerate}
\State Alice obtains outcome $\ket{\pi(m)}$ or $\ket{\widetilde{\pi(m)}}$, and announces $m\in\{0,\dots,d-1\}$. The shared state dependent on $b$, $\pi$ and $m$ is $\ket{\Xi_{\pi,m}^b}$.
\Ensure
\State Alice announces $b$ and $\pi$ and sends Bob the register $\cA$.
\State Bob projects the full state onto the expected state  $\ketbra{\Xi_{\pi,m}^b}_{\cA\cB}$ to verify that Alice's announcement was truthful. 
\end{algorithmic}
\end{algorithm}

\begin{thm}\label{thm:bc-security}
    The bit commitment protocol \ref{ame3d-qbc} is perfectly hiding and $1/d$-honest-binding, if Alice is restricted to only perform separable operations.
\end{thm}

In the following we outline the main steps of the proof of Theorem \ref{thm:bc-security}. We refer the reader to appendix \ref{apx:sec} for details. 

We first show that the commitment does not reveal any information and is always accepted when both parties are being honest. 
If Alice commits to $1$, the post-measurement state given the permutation $\pi$ and the measurement outcome $m$ is
\begin{equation*}
\ket{\Xi_{\pi, m}^1}
    =\ket{\Phi_{\pi(m)}}.
\end{equation*}
Similarly, if Alice commits to $0$, the post-measurement state is 
\begin{equation*}
\ket{\Xi_{\pi,m}^0}
=\left(\frac{1}{\sqrt{d}}\sum_{l=0}^{d-1}\om^{\pi(m)l}\ket{\Phi_l}\right).
\end{equation*}
We assume that Alice decides randomly to commit to $0$ or $1$; this means that given her measurement outcome $m$, the shared state of Alice and Bob is equally likely to be $\ket{\Xi_{\pi, m}^0}$ or $\ket{\Xi_{\pi, m}^1}$. We need to guarantee that these states are indistinguishable to Bob. Indeed, it holds that:
\begin{equation*}
    \Tr[\cA](\ketbra{\Xi_{\pi, m}^0})=\Tr[\cA](\ketbra{\Xi_{\pi, m}^1})=\frac{\1}{d}.
\end{equation*} where $\1$ is the identity matrix.
Furthermore, since for a given $\pi\in S_d$ and $b\in\{0,1\}$, $\{\ketbra{\Xi_{\pi, j}^b}\}_{j=0}^{d-1}$ forms a basis, when Alice announces her parameters $m$, $\pi$ and $b$ in Step 5, Bob can measure in that basis and finally accept Alice's opening if the outcome is indeed $m$. 

We prove \cref{thm:bc-security} by separately showing that the protocol is $\nicefrac{1}{d}$-honest-binding and perfectly hiding. For the former we show that it fulfils the conditions of \cref{thm_bin} if Alice acts honestly in the commit phase, as for the hiding we show that even if Bob sends an arbitrary state of his choice, he cannot infer any information on the bit Alice committed after the commit phase.
{ 
\begin{lem}\label{lem:bc-hon-bin}
    If Alice acts honestly in the commit phase, then Alice and Bob share a state $\ket{\Xi_{\pi,m}^b}$ that satisfies the condition \ref{cond_reduced_state} of \cref{thm_bin}, making the protocol $\nicefrac{1}{d}$-honest-binding.  
\end{lem}
}
\begin{proof}
We show that after honest commitment, Alice’s reduced state is an AME state if she committed $0$ and a separable state if she committed $1$. Indeed, if Alice selects the permutation $\pi$ and commits to $1$, the post-measurement state given outcome $m$ is:
\begin{equation*}
\ket{\Xi_{\pi, m}^1}
 =\ket{\Phi_{\pi(m)}}=
\frac{1}{\sqrt{d}}\sum_{j=0}^{d-1}\ket{y^{m,\pi}_j}\ket{j},
\end{equation*}
where $\ket{y^{m,\pi}_j}$ are product qudit states and $\{\ket{y^{m,\pi}_j}\}_{j=0}^{d-1}$ is an orthonormal set. On the other hand, if Alice commits to $0$, the post-measurement state given outcome $m$ is:
\begin{equation*}
\ket{\Xi_{\pi, m}^0}=\frac{1}{\sqrt{d}}\sum_{l=0}^{d-1}\om^{\pi(m)l}\ket{\Phi_l}=
\frac{1}{\sqrt{d}}\sum_{j=0}^{d-1}\ket{x^{m,\pi}_j}\ket{j}\,
\end{equation*}
where $\ket{x^{m,\pi}_j}$ is an $AME(3,d)$ state for every $j$, and $\{\ket{x^{m,\pi}_j}\}_{j=0}^{d-1}$ is an orthonormal set. We refer the reader to appendix \ref{apx:sec} for more detailed explanations.

It therefore follows from \cref{thm_bin}  that $p_b\leq\nicefrac{1}{d}$ for $b\in\{0,1\}$. This means that Alice can successfully reveal a different bit than the one she committed to, with probability $\nicefrac{1}{d}$.
\end{proof}

We emphasise that if Alice had not been restricted to separable operations, then her cheating probability would have been $1$, whereas for qutrits our cheating probability is already smaller or equal to $\nicefrac{1}{3}$. Next, we prove that the protocol is perfectly hiding. 
{ 
\begin{lem}\label{lem:bc-hiding}
For any state that Bob creates and sends in the AME(3,d) quantum bit commitment protocol, he does not get any information on the bit Alice commits.
\end{lem}
}
\begin{proof}
Bob can try to cheat in Step 1 of the protocol, by sending a different state that could potentially improve his probability of guessing the committed bit before the opening phase. 
Alice's measurements can be modelled as quantum channels and thus Bob's cheating probability is the probability of successfully distinguishing between Alice's two choices. 

Alice's commitment $b$ is implemented in terms of a measurement channel on the register $\anc$; she chooses a permutation $\pi\in S_d$ at random and then measures in the $\bb_{0,\pi}$ or $\bb_{1,\pi}$ basis respectively. For a given $\pi$, we can write the measurement channels as follows:
\begin{align*}
    \cM^\pi_{0}(\rho) &= \sum_{m=0}^{d-1} (\ket{m}\bra{\widetilde{\pi(m)}}\otimes\1_{\cA\cB})\rho(\ket{\widetilde{\pi(m)}}\bra{m}\otimes\1_{\cA\cB}),\\
    \cM^\pi_{1}(\rho) &= \sum_{m=0}^{d-1} (\ket{m}\bra{\pi(m)}\otimes\1_{\cA\cB})\rho(\ket{\pi(m)}\bra{m}\otimes\1_{\cA\cB}).
\end{align*}
Since $\pi$ is chosen at random, Alice's measurement channels corresponding to a $0$ and $1$ commitment are:
\begin{align*}
    \cM_{0}(\rho)=\frac{1}{d!}\sum_{\pi\in S_d}\cM^\pi_{0}(\rho),\\
    \cM_{1}(\rho)=\frac{1}{d!}\sum_{\pi\in S_d}\cM^\pi_{1}(\rho).
\end{align*}
For any $\rho$, we then have (see appendix \ref{apx:sec} for details):
\begin{equation*}
   \cM_{0}(\rho)=\cM_{1}(\rho).
\end{equation*}
Therefore the post-measurement state reveals no information on Alice's choice to Bob, no matter which state he sends Alice. The protocol is therefore perfectly hiding.
\end{proof}

\section{Conclusion}
While unconditionally secure quantum bit commitment cannot exist, { in this work we present conditions under which one could design a protocol that is binding if the committing party is restricted to separable operations. 
Specifically, we prove that if the shared post-commitment state is either an absolutely maximally entangled state or a product state after tracing out the receiver's system, then the probability of the committer successfully cheating can be bounded by $\nicefrac{1}{d}$, where $d$ is the dimension of the qudits in the communicated state.} We further strengthen our result by presenting an example protocol that is perfectly hiding and that saturates the $\nicefrac{1}{d}$ bound, if the committer is honest during the first phase of the protocol. The proposed AME(3,d) bit commitment protocol provides a foundational example of how honest-binding bit commitment can be designed within this restricted operational framework, leveraging the properties of AME states to ensure security. However in our proposed protocol, Alice can easily defer her commitment, making the protocol honest-binding, but not sum-binding.

Our results present a novel approach to circumventing the no-go theorem for unconditionally secure bit commitment in the quantum setting, emphasizing the role of operational restrictions on the committing party. \Cref{thm_bin} offers a set of conditions that if fulfilled by a protocol will lead to $\nicefrac{1}{d}$-sum-binding and perfectly hiding quantum bit commitment.  While we have established security under the honest-binding criterion,
developing protocols where the restriction to separable operations can also guarantee sum-binding security remains an important challenge which would have strong implications for the implementation of secure bit commitment schemes. 

Furthermore, the security of the proposed protocol scales with the dimension of the communicated qudits, rather than with the number of used qubits as is often the case with other cryptographic protocols, in particular \cite{salvail_quantum_1998}.  
The binding condition presented in the latter requires a large amount of qubits (roughly calculated at 7800) with a restriction of acting on maximum one qubit at a time to be close to our bound offered by qutrits, highlighting the qubit-qutrit tradeoff between scaling in numbers and in dimension. 
Qudits however remain significantly difficult to control and use with current technological resources \cite{qudits,hamiltionia-prx, hifiQutrit22}. A possible implementation of the proposed protocol using qutrits could however already be envisioned by using a ten-qubit entangled state \cite{qutrit-qubit-pra}; however, this would need to be further investigated.

Further work should also include a thorough analysis of the protocol’s robustness under realistic noise conditions. 
A starting point could be considering separable noise channels, for example depolarising channels acting on each qudit of Alice's reduced state. This would not alter our results as the effect of the noise channel would be absorbed into Alice's separable channel, leaving her cheating probability unchanged.
In contrast, correlated noise is destructive and decreases entanglement, and we believe that 
while this might increase Alice's cheating probability in one direction (switching her commitment from $0$ to $1$) it will reduce it in the other direction.
Additionally such noise could lead to Bob rejecting Alice's commitment even when she has been acting honestly. Typically this would be resolved with some mitigation mechanism (e.g. repetition of the protocol), that could impact the cheating strategies. A careful treatment of these noise effects is left for future work.

Another direction for future research is to investigate separable channels over more than two partitions and eventually return to the physically more realistic restriction of LOCC operations mentioned in the introduction. Given that the latter do not have a closed form expression \cite{chitambar_locc,watrous_theory_2018}, we anticipate the mathematical analysis to be more challenging. We expect, however, the inability of Alice to create entanglement to remain the backbone of the security proof.

Finally, we believe that the insights from this work can be extended to other two-party cryptographic primitives, such as oblivious transfer and secure multiparty computation, where similar operational restrictions might enable new secure protocols.

\section{Acknowledgements}
All authors would like to particularly thank Jamie Sikora for insightful discussions throughout the project. MR also acknowledges discussions with Andreas Winter on an alternative swap-based protocol. Finally, AP would like to thank Louis Salvail for introducing the original problem to her while on an academic stay at the University of Montreal and for discussing potential solutions. ZC and AP acknowledge support from the Hector Fellow Academy, the Emmy Noether grant from the DFG (No 41829458) and the Berlin Quantum Alliance. MR is supported by MUR via project 816000-2022-SQUID - CUP F83C22002390007 (Young Researchers). 

\bibliography{bibliography.bib}
\newpage
\onecolumngrid
\appendix
\section{Deriving the general cheating bound}\label{apx:gen}
\setlength{\parindent}{0pt}
We derive \eqref{eq:switch_schmidt_0} and \eqref{eq:switch_schmidt_1} by 
plugging \eqref{eq:schmidt_0} and \eqref{eq:schmidt_1} into \eqref{eq:fidelity_overlapsum}; indeed,
\begin{align*}
    p_0&=\sum_j|\bra{\Psi_{1}}(K_j\otimes\1_\cB)\ket{\Psi_0}|^2=\sum_j\left|\left(\sum_{i=0}^{M-1}\sqrt{\lambda_i}\bra{y_i}\bra{b_i}\right)(K_j\otimes\1_\cB)\left(\sum_{l=0}^{M-1}\sqrt{\lambda_l}\ket{x_l}\ket{b_l}\right)\right|^2\\
    &=\sum_j\left|\sum_{i=0}^{M-1}\lambda_i\bra{y_i}K_j\ket{x_i}\right|^2.
\end{align*}
Similarly we find
    $p_1=\sum_j\left|\sum_{i=0}^{M-1}\lambda_i\bra{x_i}K_j\ket{y_i}\right|^2.$

We now present detailed computations of \eqref{eq:alice_switch_bound_0} and \eqref{eq:alice_switch_bound_1}. We start with the latter. Using \eqref{eq:product_state} and \eqref{eq:entangled_states} we find
\begin{align*}
    p_1&=\sum_j\left|\sum_{i=0}^{M-1} \lambda_i\bra{x_i}K_j\ket{y_i}\right|^2=\sum_j\left|\sum_{i=0}^{M-1}\lambda_i\left(\sum_{k=0}^{N_2-1}\frac{1}{\sqrt{N_2}}\braket{e_{ik}}{v_{ji}}\braket{f_{ik}}{w_{ji}}
    \right)\right|^2\\
    &=\sum_j\frac{1}{N_2}\left|\sum_{i=0}^{M-1}\sum_{k=0}^{N_2-1}\sqrt{\lambda_i}\braket{e_{ik}}{v_{ji}}\sqrt{\lambda_i}\braket{f_{ik}}{w_{ji}}\right|^2\\
    &\leq\frac{1}{N_2}\sum_j\left(\sum_{i=0}^{M-1}\sum_{k=0}^{N_2-1}\left|\sqrt{\lambda_i}\braket{e_{ik}}{v_{ji}}\right|^2\right)\left(\sum_{s=0}^{M-1}\sum_{t=0}^{N_2-1}\left|\sqrt{\lambda_s}\braket{f_{st}}{w_{js}}\right|^2\right)\\
    &=\frac{1}{N_2}\sum_j\left(\sum_{i=0}^{M-1}\lambda_i\sum_{k=0}^{N_2-1}\left|\braket{e_{ik}}{v_{ji}}\right|^2\right)\left(\sum_{s=0}^{M-1}\lambda_s\sum_{t=0}^{N_2-1}\left|\braket{f_{st}}{w_{js}}\right|^2\right)\\
    &=\frac{1}{N_2}\sum_j\left(\sum_{i=0}^{M-1}\lambda_i\braket{v_{ji}}{v_{ji}}\right)\left(\sum_{s=0}^{M-1}\lambda_s\braket{w_{js}}{w_{js}}\right)\\
&=\frac{1}{N_2}\sum_j\sum_{i,s=0}^{M-1}\lambda_i\lambda_s\bra{v_i}K_{j1}\da K_{j1}\ket{v_i}\bra{w_s}K_{j2}\da K_{j2}\ket{w_s}\\
    &=\frac{1}{N_2}\Tr\left(\sum_j (K_{j1}\otimes K_{j2})\left(\sum_{i=0}^{M-1}\lambda_i\ketbra{v_i}\otimes\sum_{s=0}^{M-1}
    \lambda_s\ketbra{w_s}\right)(K_{j1}\otimes K_{j2})\da\right)\\ &=\frac{1}{N_2}\Tr\left(\cN\left(\sum_{i=0}^{M-1}\lambda_i\ketbra{v_i}\otimes\sum_{s=0}^{M-1}
    \lambda_s\ketbra{w_s}\right)\right)=\frac{1}{N_2}.
\end{align*}
The upper bound in the second line follows from using the Cauchy-Schwarz inequality on the double sum, treated as the sum of all the terms indexed by $i$ and $k$. The last equality holds because $\cN$ is trace-preserving. 

\allowdisplaybreaks
The case for $b=0$ is a bit more convoluted. Again we use \eqref{eq:product_state} and \eqref{eq:entangled_states}, and we denote by $\lm$ the maximum eigenvalue of $\Phi_0$ given in \eqref{eq:schmidt_0}, and, by $\Pi_K$ (where $K\in\mathbb{N}$) a projector onto a subspace of dimension of $K$.
\begin{align*}
    p_0&=\sum_j\left|\sum_{i=0}^{M-1}\lambda_i\bra{y_i}K_j\ket{x_i}\right|^2
    =\sum_j\left|\sum_{i=0}^{M-1}\sum_{k=0}^{N_2-1}\lambda_i\bra{v_i}\bra{w_i}(K_{j1}\otimes K_{j2})\frac{\ket{e_{ik}}\ket{f_{ik}}}{\sqrt{N_2}}\right|^2\\
    &=\frac{1}{N_2}\sum_j\left|\sum_{i=0}^{M-1}\sum_{k=0}^{N_2-1}\sqrt{\lambda_i}\bra{v_i}K_{j1}\ket{e_{ik}}\sqrt{\lambda_i}\bra{w_i}K_{j2}\ket{f_{ik}}\right|^2\\
    &\leq\frac{1}{N_2}\sum_j\left(\sum_{i=0}^{M-1}\sum_{k=0}^{N_2-1}|\sqrt{\lambda_i}\bra{v_i}K_{j1}\ket{e_{ik}}|^2\sum_{s=0}^{M-1}\sum_{t=0}^{N_2-1}|\sqrt{\lambda_s}\bra{w_s}K_{j2}\ket{f_{st}}|^2\right)\\
     &=\frac{1}{N_2}\sum_j\left(\sum_{i=0}^{M-1}\sum_{k=0}^{N_2-1}\lambda_i\bra{v_i}K_{j1}\ket{e_{ik}}\bra{e_{ik}}K_{j1}\da\ket{v_i}\sum_{s=0}^{M-1}\sum_{t=0}^{N_2-1}\lambda_s\bra{w_s}K_{j2}\ket{f_{st}}\bra{f_{st}}K_{j2}\da\ket{w_s}\right)\\
    &=\frac{1}{N_2}\sum_j\left(\sum_{i=0}^{M-1}\lambda_i\bra{v_i}K_{j1}\underbrace{\left(\sum_{k=0}^{N_2-1}\ket{e_{ik}}\bra{e_{ik}}\right)}_{\Pi_{N_2}}K_{j1}\da\ket{v_i}\sum_{s=0}^{M-1}\lambda_s\bra{w_s}K_{j2}\underbrace{\left(\sum_{t=0}^{N_2-1}\ket{f_{st}}\bra{f_{st}}\right)}_{\1_{N_2}}K_{j2}\da\ket{w_s}\right)\\
    &=\frac{1}{N_2}\sum_j\left(\sum_{i=0}^{M-1}\lambda_i\bra{v_i}K_{j1}\Pi_{N_2}K_{j1}\da\ket{v_i}\sum_{s=0}^{M-1}\lambda_s \bra{w_s}K_{j2}\1_{N_2}K_{j2}\da\ket{w_s}\right)\\
    &=\frac{1}{N_2}\sum_j\left(\Tr\left(\sum_{i=0}^{M-1}\lambda_i\ket{v_i}\bra{v_i}K_{j1}\Pi_{N_2}K_{j1}\da\right)\Tr\left(\sum_{s=0}^{M-1}\lambda_s \ket{w_s}\bra{w_s}K_{j2}\1_{N_2}K_{j2}\da\right)\right)\\
    &\leq\frac{\lm^2}{N_2}\sum_j\left(\Tr\left(\sum_{i=0}^{M-1}\ket{v_i}\bra{v_i}K_{j1}\Pi_{N_2}K_{j1}\da\right)\Tr\left(\sum_{s=0}^{M-1} \ket{w_s}\bra{w_s}K_{j2}\1_{N_2}K_{j2}\da\right)\right)\\
    &=\frac{\lm^2}{N_2}\sum_j\left(\Tr\left(\Pi_{M}K_{j1}\Pi_{N_2}K_{j1}\da\right)\Tr\left(\Pi_{M}K_{j2}\1_{N_2}K_{j2}\da\right)\right)\\
    &=\frac{\lm^2}{N_2}\sum_j\left(\Tr\left(K_{j1}\Pi_{N_2}K_{j1}\da\Pi_{M}\right)\Tr\left(K_{j2}\1_{N_2}K_{j2}\da\Pi_{M}\right)\right)\\  
    &=\frac{\lm^2}{N_2}\Tr\left(\sum_j(K_{j1}\otimes K_{j2})(\Pi_{N_2}\otimes\1_{N_2})(K_{j1}\otimes K_{j2})\da(\Pi_{M}\otimes\Pi_M)\right)\\
    &=\frac{\lm^2}{N_2}\Tr\left(\cN(\Pi_{N_2}\otimes\1_{N_2})(\Pi_{M}\otimes\Pi_M)\right)
    \leq\frac{\lm^2}{N_2}\Tr(\cN(\Pi_{N_2}\otimes\1_{N_2}))\\
    &=\frac{\lm^2}{N_2}\Tr(\Pi_{N_2}\otimes\1_{N_2})=\frac{\lm^2}{N_2}N_2^2=\lm^2N_2.
\end{align*}
 The first inequality again follows the Cauchy-Schwarz inequality, while the last inequality holds because $\Pi_{N_2}\otimes\1_{N_2}$ is a positive operator and $\cN$ a completely positive map, therefore mapping positive operators to positive operators.

\section{Security proof of the $AME(3,d)$-QBC protocol}\label{apx:sec}
\setlength{\parindent}{0pt}
\subsubsection{The honest case}
We first give detailed calculations to show that the protocol indeed allows for bit commitment if both parties are honest. We need to show that the states when Alice commits to $0$ or $1$ are indistinguishable to Bob. We show that Bob's reduced state is always the maximally mixed state. We first properly derive the post-commitment states. When Alice commits to $b=1$, we have
\begin{align*}
&(\ket{m}\bra{\pi(m)}\otimes \1)\ket{\Xi}=(\ket{m}\bra{\pi(m)}\otimes \1)\left(\frac{1}{\sqrt{d}}\sum_{l=0}^{d-1}\ket{l}\ket{\Phi_l}\right)= \frac{1}{\sqrt{d}}\sum_{l=0}^{d-1}\ket{m}\delta_{\pi(m)l}\ket{\Phi_l}\\
=&\frac{1}{\sqrt{d}}\ket{m}_\anc\ket{\Phi_{\pi(m)}}_{\cA\cB}.
\end{align*}
And therefore
$ \frac{(\ket{m}\bra{\pi(m)}\otimes \1)\ket{\Xi}}{\sqrt{\bra{\Xi}(\ket{\pi(m)}\bra{m}\otimes \1)(\ket{m}\bra{\pi(m)}\otimes \1)\ket{\Xi}}}=\ket{m}_\anc\ket{\Phi_{\pi(m)}}_{\cA\cB}.
$
\newline\newline
If she commits to $b=0$, we have 
\begin{align*}
&(\ket{m}\bra{\widetilde{\pi(m)}}\otimes \1)\ket{\Xi}=(\ket{m}\bra{\widetilde{\pi(m)}}\otimes \1)\left(\frac{1}{\sqrt{d}}\sum_{l=0}^{d-1}\ket{l}\ket{\Phi_l}\right)
 =\frac{1}{\sqrt{d}}\sum_{l=0}^{d-1}\ket{m}\braket{\widetilde{\pi(m)}}{l}\ket{\Phi_l}\\
  =&\frac{1}{\sqrt{d}}\sum_{l=0}^{d-1}\left(\frac{1}{\sqrt{d}}\sum_{k=0}^{d-1}\om^{\pi(m)k}\ket{m}\braket{k}{l}\right)\ket{\Phi_l}
  = \frac{1}{\sqrt{d}}\sum_{l=0}^{d-1}\left(\frac{1}{\sqrt{d}}\sum_{k=0}^{d-1}\om^{\pi(m)k}\delta_{kl}\ket{m}\right)\ket{\Phi_l}\\
    =&\frac{1}{d}\ket{m}_{\anc}\sum_{l=0}^{d-1}\om^{\pi(m)l}\ket{\Phi_l}_{\cA\cB}.
\end{align*}
Since $\om$ is a  $d$-th root of unity, it holds that
             $\sum_{i=0}^{d-1}\om^i=0,$
     which in turn implies that for all $j,k\in\{0,\dots,d-1\}$, we have:
\[  \left(\frac{1}{d}\sum_{k=0}^{d-1}\om^{-\pi(m)k}\bra{\Phi_{k}}\right)  \left(\frac{1}{d}\sum_{l=0}^{d-1}\om^{\pi(m)l}\ket{\Phi_{l}}\right)
   =\frac{1}{d^2}\sum_{l,k=0}^{d-1}\om^{\pi(m)(l-k)}\delta_{kl}=\frac{1}{d}\]
and therefore:
\begin{equation*}
\frac{(\ket{m}\bra{\widetilde{\pi(m)} }\otimes \1)\ket{\Xi_\pi}}{\sqrt{\bra{\Xi_\pi}(\ket{\widetilde{\pi(m)}}\otimes \1)(\bra{\widetilde{\pi(m)}}\otimes \1)\ket{\Xi_\pi}}}=\ket{m}_\anc\left(\frac{1}{\sqrt{d}}\sum_{l=0}^{d-1}\om^{m\pi(l)}\ket{\Phi_l}_{\cA\cB}\right).
\end{equation*}
Since the state of the ancilla register $\ket{m}_{\anc}$ is fixed after the commitment and therefore contains classical information, it can be dropped for the remainder of the analysis. We will therefore consider the states $\ket{\Xi_{\pi, m}^0}=\frac{1}{\sqrt{d}}\sum_{l=0}^{d-1}\om^{m\pi(l)}\ket{\Phi_l}$ and $\ket{\Xi_{\pi, m}^1}=\ket{\Phi_{\pi(m)}}$.

We also want Bob's reduced states to be indistinguishable and maximally mixed. Bob's reduced state when Alice commits $1$ is:
\begin{align*}
    &\Tr[\cA](\ketbra{\Xi_{\pi, m}^1})
    =\Tr[\cA](\ketbra{\Phi_{\pi(m)}})\\
    =&\Tr[\cA]\left(\frac{1}{\sqrt{d}}\sum_{l=0}^{d-1}\om^{\pi(m)l}\ket{lll}\ket{l+\pi(m)}\right)\left(\frac{1}{\sqrt{d}}\sum_{j=0}^{d-1}\om^{-\pi(m)j}\bra{jjj}\bra{j+\pi(m)}\right)\\
    =&\frac{1}{d}\sum_{j,l=0}^{d-1}\om^{\pi(m)(l-j)}\delta_{lj}\ket{l+\pi(m)}\bra{j+\pi(m)}
    =\frac{1}{d}\sum_{j=0}^{d-1}\ket{j+\pi(m)}\bra{j+\pi(m)}=\frac{1}{d}\sum_{j'=0}^{d-1}\ket{j'}\bra{j'}
    =\frac{\1_\cB}{d}.
\end{align*}
Similarly when she commits $0$, Bob's reduced state is:
\begin{align*}
    &\Tr[\cA](\ketbra{\Xi_{\pi, m}^0})
    =\Tr[\cA]\left[\left(\frac{1}{\sqrt{d}}\sum_{k=0}^{d-1}\om^{\pi(m)k}\ket{\Phi_k}\right)\left(\frac{1}{\sqrt{d}}\sum_{j=0}^{d-1}\om^{-\pi(m)j}\bra{\Phi_j}\right)\right]\\
    =&\frac{1}{d}\sum_{j,k=0}^{d-1}\om^{\pi(m)(k-j)}\Tr[\cA]\left[\ket{\Phi_k}\bra{\Phi_j}\right]\\
    =&\frac{1}{d}\sum_{j,k=0}^{d-1}\om^{\pi(m)(k-j)}\Tr[\cA]\left[\left(\frac{1}{\sqrt{d}}\sum_{s=0}^{d-1}\om^{ks}\ket{sss}\ket{s+k}\right)\left(\frac{1}{\sqrt{d}}\sum_{t=0}^{d-1}\om^{-jt}\bra{ttt}\bra{t+j}\right)\right]\\
    =&\frac{1}{d^2}\sum_{j,k=0}^{d-1}\om^{\pi(m)(k-j)}\sum_{s,t=0}^{d-1}\om^{ks-jt}\delta_{st}\ket{s+k}\bra{t+j}
    =\frac{1}{d^2}\sum_{j,k=0}^{d-1}\om^{\pi(m)(k-j)}\sum_{s=0}^{d-1}\om^{s(k-j)}\ket{s+k}\bra{s+j}.
\end{align*}
We relabel the sum with $j'=s+j$ and $k'=s+k$, and note that $k'-j'=k-j$. Summing over roots of unity again yields:
    \begin{align*}
    &\frac{1}{d^2}\sum_{j,k=0}^{d-1}\om^{\pi(m)(k-j)}\sum_{s=0}^{d-1}\om^{s(k-j)}\ket{s+k}\bra{s+j}
    =\frac{1}{d^2}\sum_{j',k'=0}^{d-1}\om^{\pi(m)(k'-j')}\left(\sum_{s=0}^{d-1}\om^{s(k'-j')}\right)\ket{k'}\bra{j'}\\
    =&\frac{1}{d^2}\sum_{j',k'=0}^{d-1}\om^{\pi(m)(k'-j')}d\delta_{j'k'}\ket{k'}\bra{j'}
    =\frac{1}{d}\sum_{j'=0}^{d-1}\ketbra{j'}
    =\frac{\1_\cB}{d}.
\end{align*}
The states are therefore indistinguishable to Bob in the honest case.\\

We also want to prove that the sets $\{\ket{\Xi_{\pi, j}^0\}}_{j=0}^{d-1}$ and $\{\ket{\Xi_{\pi, j}^1}\}_{j=0}^{d-1}$ are orthonormal sets. The definition \eqref{def:phi_l} of $\Phi_l$ implies that $\{\ket{\Xi_{\pi, j}^1}\}_{j=0}^{d-1}$ is an orthonormal set. For $\{\ket{\Xi_{\pi, j}^0\}}_{j=0}^{d-1}$ we have
\begin{align*}
    &\braket{\Xi_{\pi, l}^0}{\Xi_{\pi, m}^0}=\left(\frac{1}{\sqrt{d}}\sum_{k=0}^{d-1}\om^{-\pi(l)k}\bra{\Phi_k}\right)\left(\frac{1}{\sqrt{d}}\sum_{j=0}^{d-1}\om^{\pi(m)j}\ket{\Phi_j}\right)\\
    =&\frac{1}{d}\sum_{k,j=0}^{d-1}\om^{\pi(m)j-\pi(l)k}\braket{\Phi_k}{\Phi_j}=\frac{1}{d}\sum_{k,j=0}^{d-1}\om^{\pi(m)j-\pi(l)k}\delta_{kj}
    =\frac{1}{d}\sum_{k=0}^{d-1}\om^{(\pi(m)-\pi(l))k}=\delta_{lm}.
\end{align*}

\subsection{Honest-binding}

When Alice commits to $1$, the shared state is:
\begin{align*}
&\ket{\Xi_{\pi, m}^1}=\ket{\Phi_{\pi(m)}}=\frac{1}{\sqrt{d}}\sum_{j=0}^{d-1}\om^{j\pi(m)}\ket{jjj}\ket{j+\pi(m)}\\
=&\frac{1}{\sqrt{d}}\sum_{j=0}^{d-1}\underbrace{\om^{(j-\pi(m))\pi(m)}\ket{(j-\pi(m))(j-\pi(m))(j-\pi(m))}}_{\ket{y^{m,\pi}_j}}\ket{j}\\
=&\frac{1}{\sqrt{d}}\sum_{j=0}^{d-1}\ket{y^{m,\pi}_j}\ket{j},
\end{align*}
where $\ket{y^{m,\pi}_j}$ are product qudit states that form an orthonormal set.\\

The post-commitment state after Alice commits to $0$ is:
\begin{align*}
&\ket{\Xi_{\pi, m}^0}=\frac{1}{\sqrt{d}}\sum_{l=0}^{d-1}\om^{\pi(m)l}\ket{\Phi_l}
=\frac{1}{\sqrt{d}}\sum_{l=0}^{d-1}\om^{\pi(m)l}\sum_{j=0}^{d-1}\frac{\om^{lj}}{\sqrt{d}}\ket{jjj}\ket{j+l}
=\frac{1}{d}\sum_{j=0}^{d-1}\sum_{l=0}^{d-1}\om^{\pi(m)l+jl}\ket{jjj}\ket{j+l}\\
=&\frac{1}{\sqrt{d}}\sum_{j=0}^{d-1}\left(\underbrace{\frac{1}{\sqrt{d}}\sum_{l=0}^{d-1}\om^{\pi(m)l+(j-l)l}\ket{(j-l)(j-l)(j-l)}}_{\ket{x^{m,\pi}_j}}\right)\ket{j}
=\frac{1}{\sqrt{d}}\sum_{j=0}^{d-1}\ket{x^{m,\pi}_j}\ket{j}\label{eq:post_measurement_ame},
\end{align*}
One can easily show that states of the form $\frac{1}{\sqrt{d}}\sum_{j=0}^{d-1}\om^{k_j}\ket{jjj}$ are AME, and therefore $\ket{x^{m,\pi}_j}$ are $AME(3,d)$ states that form an orthonormal set. \\

\subsection{Hiding}
We will now prove that for any $\rho$, it holds that $\cM_{0}(\rho)=\cM_{1}(\rho)$ and therefore the commitment does not reveal any information to Bob. To do this we will use the following fact: If $S_d$ is the group of all permutations over $d$ elements, for all $j,k\in \{1,\dots,d\}$ we have $|\{\sigma\in S_d : \sigma(j)=k\}|=(d-1)!$. It then follows:

\begin{align*}
&\cM_{0}(\rho)=\frac{1}{d!}\sum_{\pi\in S_d}\sum_{m=0}^{d-1} (\ket{m}\bra{\widetilde{\pi(m)}}\otimes\1_{\cB})\rho(\ket{\widetilde{\pi(m)}}\bra{m}\otimes\1_{\cB})\\
=&\frac{1}{d!}\sum_{m=0}^{d-1} \sum_{\pi\in S_d}\sum_{k,l=0}^{d-1}\frac{1}{d}\om^{\pi(m)(k-l)}(\ket{m}\bra{k}\otimes\1_{\cB})\rho(\ket{l}\bra{m}\otimes\1_{\cB})\\
=&\frac{1}{dd!}\sum_{m=0}^{d-1} \sum_{k,l=0}^{d-1}\sum_{\pi\in S_d}\om^{\pi(m)(k-l)}(\ket{m}\bra{k}\otimes\1_{\cB})\rho(\ket{l}\bra{m}\otimes\1_{\cB})\\
=&\frac{1}{dd!}\sum_{m=0}^{d-1} \sum_{k,l=0}^{d-1}(d-1)!\sum_{q=0}^{d-1}\om^{q(k-l)}(\ket{m}\bra{k}\otimes\1_{\cB})\rho(\ket{l}\bra{m}\otimes\1_{\cB})\\
=&\frac{1}{d^2}\sum_{m=0}^{d-1} \sum_{k,l=0}^{d-1}d\delta_{kl}(\ket{m}\bra{k}\otimes\1_{\cB})\rho(\ket{l}\bra{m}\otimes\1_{\cB})\\
=&\frac{1}{d}\sum_{m=0}^{d-1} \sum_{k=0}^{d-1}(\ket{m}\bra{k}\otimes\1_{\cB})\rho(\ket{k}\bra{m}\otimes\1_{\cB}).\\
\end{align*}
Similarly we compute: 
\begin{align*}
&\cM_{1}(\rho)=\frac{1}{d!}\sum_{\pi\in S_d}\sum_{m=0}^{d-1} (\ket{m}\bra{\pi(m)}\otimes\1_{\cB})\rho(\ket{\pi(m)}\bra{m}\otimes\1_{\cB})\\
=&\frac{1}{d!}\sum_{m=0}^{d-1}\sum_{\pi\in S_d}(\ket{m}\bra{\pi(m)}\otimes\1_{\cB})\rho(\ket{\pi(m)}\bra{m}\otimes\1_{\cB})\\
=&\frac{1}{d!}\sum_{m=0}^{d-1} (d-1)!\sum_{k=0}^{d-1}(\ket{m}\bra{k}\otimes\1_{\cB})\rho(\ket{k}\bra{m}\otimes\1_{\cB})\\
=&\frac{1}{d}\sum_{m=0}^{d-1}\sum_{k=0}^{d-1}(\ket{m}\bra{k}\otimes\1_{\cB})\rho(\ket{k}\bra{m}\otimes\1_{\cB}).\\
\end{align*}

\end{document}